\documentclass[11pt]{article}

\usepackage{amssymb}
\usepackage{amsfonts}
\usepackage{amsmath}
\usepackage{amsthm}
\usepackage[english]{babel}
\usepackage{latexsym}
\usepackage{color}
\usepackage{enumerate}
\usepackage{nicefrac}
\usepackage{mathrsfs}
\usepackage{comment}
\usepackage[hmargin=2cm,vmargin=2cm]{geometry}
\usepackage{bm}
\usepackage{bbm}

\usepackage{centernot}

\usepackage[affil-it]{authblk}
\theoremstyle{plain}
\newtheorem{theorem}{Theorem}[section]
\newtheorem{lemma}[theorem]{Lemma}
\newtheorem{proposition}[theorem]{Proposition}
\newtheorem{corollary}[theorem]{Corollary}
\newtheorem{example}[theorem]{Example}

\theoremstyle{definition}
\newtheorem{definition}[theorem]{Definition}
\newtheorem{assumption}[theorem]{Assumption}
\newtheorem{remark}[theorem]{Remark}

\theoremstyle{remark}

\numberwithin{equation}{section}

\newcommand{\ba}{\begin{array}{ll}}
\newcommand{\bal}{\begin{array}{ll}}
\newcommand{\ea}{\end{array}}

\newcommand{\E}{\mathbb{E}}

\newcommand{\probp}{\mathbb{P}}

\newcommand{\R}{\mathbb{R}}
\newcommand{\N}{\mathbb{N}}

\newcommand{\cE}{{\mathcal{E}}}
\newcommand{\cF}{{\mathcal{F}}}

\newcommand{\cS}{{\mathcal{S}}}
\newcommand{\cA}{\mathcal{A}}

\newcommand{\cL}{\mathcal{L}}
\newcommand{\cM}{\mathcal{M}}

\newcommand{\cX}{{\mathcal{X}}}

\newcommand{\one}{\mathbbm 1}

\newcommand{\dom}{\mathop {\rm dom}\nolimits}

\newcommand{\Span}{\mathop{\rm span}\nolimits}

\def\keywords{\vspace{.5em}
{\noindent\textbf{Keywords}:\,\relax%
}}

\makeatletter
\def\@fnsymbol#1{\ensuremath{\ifcase#1\or 1\or 2\or 3\or 4\or 5\or 6\or 7\or 8\else\@ctrerr\fi}}
\makeatother

\begin{document}

\title{Law-invariant functionals that collapse to the mean}

\author{\sc{Fabio Bellini}}
\affil{Department of Statistics and Quantitative Methods\\
University of Milano-Bicocca, Italy\\
\texttt{\normalsize fabio.bellini@unimib.it}}
\author{\sc{Pablo Koch-Medina},\,\,\sc{Cosimo Munari}}
\affil{Center for Finance and Insurance and Swiss Finance Institute\\
University of Zurich, Switzerland\\
\texttt{\normalsize pablo.koch@bf.uzh.ch},\,\,\texttt{\normalsize cosimo.munari@bf.uzh.ch}}
\author{\sc{Gregor Svindland}}
\affil{Institute of Probability and Statistics and House of Insurance\\
Leibniz University Hannover, Germany\\
\texttt{\normalsize gregor.svindland@insurance.uni-hannover.de}}

\date{\today}

\maketitle

\begin{abstract}
\noindent We discuss when law-invariant convex functionals ``collapse to the mean''. More precisely, we show that, in a large class of spaces of random variables and under mild semicontinuity assumptions, the expectation functional is, up to an affine transformation, the only law-invariant convex functional that is linear along the direction of a nonconstant random variable with nonzero expectation. This extends results obtained in the literature in a bounded setting and under additional assumptions on the functionals. We illustrate the implications of our general results for pricing rules and risk measures.
\end{abstract}


\parindent 0em \noindent


\section{Introduction}

In a well-known paper Wang et al.\ \cite{WangYoungPanjer1997}, the authors describe an axiomatic approach to insurance pricing and provide a representation of admissible pricing rules in terms of Choquet integrals. One of the key axioms put forward is law invariance, stipulating that prices depend on the contracts' payoffs only through their probability distribution with respect to the ``physical'' probability measure. At the end of that paper, it is pointed out that law-invariant pricing rules based on Choquet integrals could also be used to harmonize the pricing of insurance products and financial derivatives. It is, however, not difficult to see that law invariance of the pricing functional cannot be expected to hold in general. For instance, the Fundamental Theorem of Asset Pricing asserts that, under suitable conditions, in a financial market that is frictionless and free of arbitrage opportunities, prices can be essentially expressed as expectations with respect to a ``risk-neutral'' probability measure. It is with respect to such a probability measure that prices in this market are law invariant. Hence, for financial market prices to exhibit law invariance with respect to the ``physical'' probability measure, the ``physical'' and the ``risk-neutral'' measures would have to coincide. This is, however, never the case with the sole exception of a market in which the expected returns under the ``physical'' measure is the same for all assets.

\smallskip

Prompted by the attempts in Wang \cite{Wang2000,Wang2002} to carry out the harmonization suggested in Wang et al.~\cite{WangYoungPanjer1997} by means of law-invariant pricing rules, Castagnoli et al.\ \cite{CastagnoliMaccheroniMarinacci2004} show that postulating the law invariance of pricing functionals is questionable also in a more general setting than that of frictionless financial markets. This was accomplished by proving that the expectation under the ``physical'' probability measure is the only pricing functional defined on the space of bounded payoffs that is law invariant, sublinear, increasing, and comonotonic (properties satisfied by the pricing rules considered in Wang \cite{Wang2000,Wang2002}), and under which every riskless payoff and at least one risky payoff are priced in a frictionless way. This ``collapse to the mean'' was improved in Frittelli and Rosazza Gianin~\cite{FrittelliRosazza2005} by replacing sublinearity with convexity and by dropping comonotonicity. We note though that, strictly speaking, these results cannot be directly applied to the setting of Wang \cite{Wang2000,Wang2002} because the payoffs considered there are not necessarily bounded. A detailed discussion of the results in Castagnoli et al.\ \cite{CastagnoliMaccheroniMarinacci2004} and Frittelli and Rosazza Gianin~\cite{FrittelliRosazza2005} and how they relate to ours is given at the beginning of Section~5.

\smallskip

The preceding discussion raises the question of whether the ``collapse to the mean'' remains valid for a \textit{wider range of spaces of random variables} and for a \textit{larger class of law-invariant functionals}. In this note, we allow the model space $\cX$ to belong to a fairly general class of locally-convex spaces consisting of integrable random variables and containing all bounded random variables. In Theorem \ref{theo: collapse} we prove that, under suitable lower semicontinuity properties (which are always satisfied in the setting of Castagnoli et al.\ \cite{CastagnoliMaccheroniMarinacci2004} and Frittelli and Rosazza Gianin \cite{FrittelliRosazza2005}), the expectation functional is, up to an affine transformation, the only law-invariant convex functional $\varphi:\cX\to(-\infty,\infty]$ that is linear along a nonconstant random variable $Z$ with nonzero expectation. The strategy we follow differs from the one used in the referenced papers and relies on the identification of an inherent tension between law invariance and linearity that sheds new light into why law invariance has such strong structural implications. The key observation, established in Lemma~4.4, is that the set of random variables that have the same distribution as $Z$ spans a dense subspace of $\cX$. As a result, linearity along $Z$ together with law invariance forces linearity on this dense subspace. The lower semicontinuity assumption then implies that $\varphi$ is linear on the entire space. The result follows by noting that the only continuous linear functionals that are law invariant are multiples of the expectation functional. This new version of the ``collapse to the mean'' has natural applications to insurance pricing rules, which is our motivating problem, as well as to risk measures. In particular, it provides a rigorous argument for why, contrary to what was claimed in Wang \cite{Wang2000,Wang2002}, law-invariant insurance pricing rules cannot be expected to reproduce prices in a frictionless and arbitrage-free financial market.

\smallskip

The note is organized as follows. In Section 2 we introduce the setting together with the necessary notation and terminology. In Section 3 we show that convex functionals that are lower semicontinuous and linear along a given direction enjoy the stronger property of being translation invariant along the same direction. In Section 4 we establish our main result on the ``collapse to the mean''. Some applications of our result are discussed in Section 5.


\section{Setting, notation, terminology}
\label{sect: setup}

Let $(\Omega,\cF,\probp)$ be a nonatomic probability space. We denote by $L^0$ the set of equivalence classes of random variables, i.e.\ Borel measurable functions $X:\Omega\to\R$, with respect to almost-sure equality under $\probp$. In line with standard practice, we do not distinguish explicitly between an element of $L^0$ and any of its representatives. In particular, the elements of $\R$ are naturally identified with random variables that are almost-surely constant. For two random variables $X,Y\in L^0$ we write $X\sim Y$ whenever $X$ and $Y$ have the same probability law under $\probp$. The expectation under $\probp$ is denoted by $\E_\probp$. The standard Lebesgue spaces are denoted by $L^p$ for $p\in[1,\infty]$. We say that a set $\cX\subset L^0$ is law invariant (under $\probp$) if $X\in\cX$ for every $X\in L^0$ such that $X\sim Y$ for some $Y\in\cX$.

\smallskip

\begin{assumption}
\label{assumption one}
We denote by $(\cX,\cX^\ast)$ a pair of law-invariant vector subspaces of $L^1$ containing $L^\infty$. We assume that $XY\in L^1$ for all $X\in\cX$ and $Y\in\cX^\ast$ and denote by $\sigma(\cX,\cX^\ast)$ the weakest linear topology on $\cX$ with respect to which, for every $Y\in\cX^\ast$, the linear functional on $\cX$ given by $X\mapsto\E_\probp[XY]$ is continuous.
\end{assumption}

\smallskip

\begin{remark}
(i) Note that, under our assumptions, $\sigma(\cX,\cX^\ast)$ is not metrizable.\footnote{In general, weak topologies can be metrizable, but not in our setting. Using the argument in the proof of the implication ``(4) $\implies$ (1)'' in Theorem 6.26 in \cite{AliprantisBorder2006}, one can show that metrizability of $\cX$ under $\sigma(\cX,\cX^\ast)$ would imply that $\cX^\ast$ can be written as the countable union of finite dimensional subspaces. Being a subspace of $\cX^\ast$, $L^\infty$ would also have this property. Baire's Lemma (Theorem~3.46 in \cite{AliprantisBorder2006}) would then imply that $L^\infty$ is finite dimensional, a contradiction. Hence, $\sigma(\cX,\cX^\ast)$ is not metrizable.} As a result, in general, one needs to work with nets instead of sequences. Recall that a net $(X_\alpha)\subset\cX$ converges to an element $X\in\cX$ with respect to the topology $\sigma(\cX,\cX^\ast)$ if and only if $\E_\probp[X_\alpha Y]\to\E_\probp[XY]$ for every $Y\in\cX^\ast$.

\smallskip

(ii) Note that for every nonzero $X\in\cX$ there exists $Y\in\cX^\ast$, namely either $Y=\one_{\{X>0\}}$ or $Y=\one_{\{X<0\}}$ (which belong to $\cX^\ast$ because they are bounded), such that $\E_\probp[XY]\neq 0$. Similarly, for every nonzero $Y\in\cX^\ast$ there exists $X\in\cX$ such that $\E_\probp[XY]\neq 0$. Hence, $(\cX,\cX^\ast)$ is a dual pair. In particular, Theorem 5.93 in Aliprantis and Border \cite{AliprantisBorder2006} implies that, endowed with $\sigma(\cX,\cX^\ast)$, the space $\cX$ is a locally-convex Hausdorff topological vector space whose topological dual can be identified with $\cX^\ast$.
\end{remark}

\smallskip

We next highlight that the class of spaces we consider is sufficiently general to accommodate virtually all Banach spaces encountered in applications as long as their dual can be identified with a space of integrable random variables. As is usual in the literature on law invariance, this rules out  $L^\infty$ with its norm dual which consists of signed finitely additive measures.

\smallskip

\begin{example}[\textbf{Orlicz Spaces}]
Let $\Phi:[0,\infty)\to[0,\infty]$ be an Orlicz function, i.e.\ a convex, left-continuous, increasing function which is finite on a right neighborhood of zero and satisfies $\Phi(0)=0$. The conjugate of $\Phi$ is the function $\Phi^\ast:[0,\infty)\to[0,\infty]$ defined by
\[
\Phi^\ast(u) := \sup_{t\in[0,\infty)}\{tu-\Phi(t)\}.
\]
Note that $\Phi^\ast$ is also an Orlicz function. For every $X\in L^0$ define the Luxemburg norm by
\[
\|X\|_\Phi := \inf\left\{\lambda\in(0,\infty) \,; \ \E\left[\Phi\left(\frac{|X|}{\lambda}\right)\right]\leq1\right\}.
\]
The corresponding Orlicz space is given by
\[
L^\Phi:=\{X\in L^0 \,; \ \|X\|_\Phi<\infty\}.
\]
The heart of $L^\Phi$ is the space
\[
H^ \Phi:=\left\{X\in L^\Phi \,; \ \forall \lambda\in(0,\infty) \,:\, \E\left[\Phi\left(\frac{|X|}{\lambda}\right)\right]<\infty\right\}.
\]
The classical Lebesgue spaces are special examples of Orlicz spaces. Indeed, if $\Phi(t)=t^p$ for $p\in[1,\infty)$ and $t\in[0,\infty)$, then $L^\Phi=H^\Phi=L^p$ and the Luxemburg norm coincides with the usual $p$ norm. Moreover, if we set $\Phi(t)=0$ for $t\in[0,1]$ and $\Phi(t)=\infty$ otherwise, then we have $L^\Phi=L^\infty$ and the Luxemburg norm coincides with the usual $L^\infty$-norm. Note that, in this case, $H^\Phi=\{0\}$.

\smallskip

In our nonatomic setting, $L^\Phi=H^\Phi$ if and only if $\Phi$ satisfies the $\Delta_2$ condition, i.e.\ there exist $s\in(0,\infty)$ and $k\in(0,\infty)$ such that $\Phi(2t)<k\Phi(t)$ for every $t\in[s,\infty)$. A well-known example of a nontrivial $H^\Phi$ with $H^\Phi\neq L^\phi$ is obtained by setting $\Phi(t)=\exp(t)-1$ for $t\in[0,\infty)$.

\smallskip

In general, the norm dual of $L^\Phi$ cannot be identified with a subspace of $L^0$. However, if $\Phi$ is finite valued (so that $H^\Phi\neq\{0\}$), the norm dual of $H^\Phi$ can always be identified with $L^{\Phi^\ast}$. For the case $L^p$, for $p\in[1,\infty)$, this is simply the well-known identification of the norm dual of $L^p$ with $L^{\frac{p}{p-1}}$ (with the usual convention $\frac10:=\infty$). For more details on Orlicz spaces we refer to Edgar and Sucheston \cite{EdgarSucheston1992}.

\smallskip

The pair $(\cX,\cX^\ast)$ with $\cX=L^\Phi$ and $\cX^\ast\in\{L^{\Phi^\ast},H^{\Phi^\ast},L^\infty\}$ satisfies Assumption~\ref{assumption one}.
\end{example}

\smallskip

In the following definition we introduce the necessary terminology for functionals.

\begin{definition}
Let $\varphi:\cX\to(-\infty,\infty]$ be a functional. The {\em domain} of  $\varphi$ is the set
\[
\dom(\varphi) := \{X\in\cX \,; \ \varphi(X)<\infty\}.
\]
We say that the functional $\varphi$ is:
\begin{enumerate}[(1)]
  \item {\em proper} if $\dom(\varphi)$ is nonempty.
  \item {\em convex} if $\varphi(\lambda X+(1-\lambda)Y)\leq\lambda\varphi(X)+(1-\lambda)\varphi(Y)$ for all $X,Y\in\cX$ and $\lambda\in[0,1]$.
  \item {\em positively homogeneous} if $\varphi(0)=0$ and $\varphi(\lambda X)=\lambda\varphi(X)$ for all $X\in\cX$ and $\lambda\in(0,\infty)$.
  \item {\em sublinear} if it is both convex and positively homogeneous.
  \item {\em increasing} if $\varphi(X)\geq\varphi(Y)$ for all $X,Y\in\cX$ such that $X\geq Y$.
  \item {\em decreasing} if $\varphi(X)\leq\varphi(Y)$ for all $X,Y\in\cX$ such that $X\geq Y$.
  \item {\em law invariant} if $\varphi(X)=\varphi(Y)$ for all $X,Y\in\cX$ such that $X\sim Y$.
  \item {\em $\sigma(\cX,\cX^\ast)$-lower semicontinuous} if for all nets $(X_\alpha)\subset\cX$ and $X\in\cX$ we have
\[
X_\alpha\xrightarrow{\sigma(\cX,\cX^\ast)}X \ \implies \ \varphi(X)\leq\liminf_{\alpha}\varphi(X_\alpha).
\]
  \item {\em norm-lower semicontinuous} if for all sequences $(X_n)\subset\cX$ and $X\in\cX$ we have
\[
X_n\xrightarrow{\|\cdot\|}X \ \implies \ \varphi(X)\leq\liminf_{n\to\infty}\varphi(X_n)
\]
provided that $\cX$ is equipped with a norm $\|\cdot\|$.
\end{enumerate}
Finally, we say that the functional $\varphi$ satisfies:
\begin{enumerate}[(10)]
  \item the {\em Fatou property} if for all sequences $(X_n)\subset\cX$ and $X\in\cX$ we have
\[
X_n\xrightarrow{a.s.}X, \ \sup_{n\in\N}|X_n|\in\cX \ \implies \ \varphi(X)\leq\liminf_{n\to\infty}\varphi(X_n).
\]
\end{enumerate}
\end{definition}

\smallskip

To a proper functional $\varphi:\cX\to(-\infty,\infty]$ we associate the \textit{dual functional} $\varphi^\ast:\cX^\ast\to(-\infty,\infty]$ defined by
\[
\varphi^\ast(Y) := \sup_{X\in\cX}\{\E_\probp[XY]-\varphi(X)\}.
\]
Note that $\varphi^\ast$ is well defined and does not attain the value $-\infty$ because $\varphi$ is proper. The next proposition records the well-known dual representation of convex and lower semicontinuous functionals; see, e.g., Theorem 2.3.3 in Z\u{a}linescu \cite{Zalinescu2002}.

\begin{proposition}
\label{prop: fenchel moreau}
Let $\varphi:\cX\to(-\infty,\infty]$ be proper, convex, and $\sigma(\cX,\cX^\ast)$-lower semicontinuous. Then, for every $X\in\cX$ we have
\[
\varphi(X) = \sup_{Y\in\cX^\ast}\{\E_\probp[XY]-\varphi^\ast(Y)\} = \sup_{Y\in\dom(\varphi^\ast)}\{\E_\probp[XY]-\varphi^\ast(Y)\}.
\]
\end{proposition}

\medskip

The next example serves to highlight that requiring $\sigma(\cX,\cX^\ast)$-lower semicontinuity for convex and law-invariant functionals is not as restrictive as it may seem at first sight since, on standard spaces, $\sigma(\cX,\cX^\ast)$-lower semicontinuity for this type of functionals is implied by fairly common continuity properties.

\begin{example}[\textbf{Orlicz Spaces}]
\label{example: fatou}
The following results can be found in Proposition~2.5 in Bellini et al \cite{BelliniKochMunariSvindland2020}, which merely summarizes results from the literature (Jouini et al.\ \cite{JouiniSchachermayerTouzi2006}, Svindland \cite{Svindland2010}, and Gao et al.\ \cite{GaoLeungMunariXanthos2018}. We also refer to Leung and Tantrawan \cite{LeungTantrawan2020} for abstract results beyond the Orlicz setting).

\smallskip

If $\cX$ is a general Orlicz space $L^\Phi$ and $\varphi:\cX\to(-\infty,\infty]$ is a proper, convex, and law invariant functional, then the following statements are equivalent:
\begin{enumerate}
\item[(a)] $\varphi$ is $\sigma(\cX,L^\infty)$-lower semicontinuous.
\item[(b)] $\varphi$ satisfies the Fatou property.
\end{enumerate}
If $\cX$ is either $L^\infty$ or an Orlicz heart $H^\Phi$ for a finite Orlicz function $\Phi$ (in particular, any $L^p$ with $1\leq p<\infty$), then (a) is also equivalent to:
\begin{enumerate}
\item[(c)] $\varphi$ is norm lower semicontinuous.
\end{enumerate}
The example given in Remark~5.6 in Gao et al.\ \cite{GaoLeungMunariXanthos2018} shows that, for a general Orlicz space, norm lower semicontinuity does not always imply $\sigma(\cX,L^\infty)$ lower semicontinuity. If $\varphi$ is additionally increasing, then (a) is also equivalent to:
\begin{enumerate}
  \item[(d)] $\varphi$ is continuous from below, i.e.\ for every increasing sequence $(X_n)\subset\cX$ and every $X\in\cX$ we have
\[
X_n\xrightarrow{a.s.}X \ \implies \ \varphi(X_n)\to\varphi(X).
\]
\end{enumerate}
Clearly, in all these cases, $\varphi$ is also $\sigma(\cX,\cX^\ast)$-lower semicontinuous.
\end{example}


\section{Affinity and translation invariance}

The goal of this short section is to show the link between two properties of functionals that will play a key role in our main result in the next section, namely affinity and translation invariance. The functionals considered in this section are not required to be law invariant. Throughout we assume that $(\cX,\cX^\ast)$ is a pair satisfying Assumption~\ref{assumption one}. For a set $\cS\subset\cX$ we denote by $\Span(\cS)$ the smallest linear subspace of $\cX$ containing $\cS$. If $\cS=\{Z\}$ for some $Z\in\cX$, then we simply write $\Span(Z)$.

\begin{definition}
Let $\cM$ be a linear subspace of $\cX$. We say that a functional $\varphi:\cX\to(-\infty,\infty]$ is:
\begin{enumerate}[(1)]
\item {\em affine along $\cM$} if $\cM\subset\dom(\varphi)$ and the functional on $\cM$ given by $Z\mapsto\varphi(Z)-\varphi(0)$ is linear. If $\cM=\Span(Z)$ for some $Z\in\cX$, then we simply say that $\varphi$ is {\em affine along $Z$}. In this case, there exists $a\in\R$ such that for every $m\in\R$
\[
\varphi(mZ) = am+\varphi(0).
\]
\item {\em translation invariant along $\cM$} if $\varphi$ is affine along $\cM$ and for all $X\in\cX$ and $Z\in\cM$
\[
\varphi(X+Z) = \varphi(X)+\varphi(Z)-\varphi(0).
\]
If $\cM=\Span(Z)$ for some $Z\in\cX$, then we simply say that $\varphi$ is {\em translation invariant along $Z$}. In this case, there exists $a\in\R$ such that for all $X\in\cX$ and $m\in\R$
\[
\varphi(X+mZ) = \varphi(X)+am.
\]
\end{enumerate}
In both cases we have $a=\varphi(Z)-\varphi(0)$.
\end{definition}

\smallskip

\begin{remark}
\label{rem: translation invariance and affinity}
Let $\cS\subset\cX$ and assume that $\varphi:\cX\to(-\infty,\infty]$ is translation invariant along every element of $\cS$. Then, $\varphi$ is translation invariant along $\Span(\cS)$. In particular, $\varphi$ is affine on $\Span(\cS)$. However, note that $\varphi$ need not be affine along $\Span(\cS)$ if it is affine along every element of $\cS$. Clearly, the only functionals that are translation invariant along $\cX$ are those that are affine on $\cX$.
\end{remark}

\smallskip

By definition, translation invariance implies affinity. As shown by the next example, the converse implication does not hold in general even if we assume that $\varphi$ is convex.

\begin{example}
\label{example convex and affine}
Assume $W,Z\in L^1$ are linearly independent and define a functional $\varphi:L^1\to(-\infty,\infty]$ by
\[
\varphi(X)=
\begin{cases}
0 & \mbox{if $X=\alpha W+\beta Z$ for some $\alpha,\beta\in\R$ with $\alpha<1$},\\
\beta^2 & \mbox{if $X=W+\beta Z$ for some $\beta\in\R$},\\
\infty & \mbox{otherwise}.
\end{cases}
\]
It is not difficult to verify that $\varphi$ is convex and also affine along $Z$. However, $\varphi$ is not translation invariant along $Z$ because there exists no $a\in\R$ such that $m^2=\varphi(W+mZ)=\varphi(W)+am=am$ for every $m\in\R$.
\end{example}

\smallskip

There are two notable classes of functionals for which affinity does imply translation invariance. The first is the class of sublinear functionals.

\begin{proposition}
\label{prop: equivalence affinity under sublinearity}
Let $\varphi:\cX\to(-\infty,\infty]$ be sublinear and $\cS\subset\cX$. If $\varphi$ is affine along every element of $\cS$, then it is translation invariant along $\Span(\cS)$.
\end{proposition}
\begin{proof}
Recall that $\varphi(0)=0$ by sublinearity and note that for every fixed $Z\in\cS$ the functional $\varphi$ is linear on $\Span(Z)$ by affinity. Hence, for every $X\in\cX$ we have
\begin{align*}
\varphi(X+Z)
&\le
\varphi(X)+\varphi(Z)\\
&=
\varphi(X+Z-Z)+\varphi(Z)\\
&\le
\varphi(X+Z)+\varphi(-Z)+\varphi(Z)\\
&=
\varphi(X+Z)
\end{align*}
by sublinearity. This shows that $\varphi$ is translation invariant along every element of $\cS$. Remark~\ref{rem: translation invariance and affinity} now implies that $\varphi$ is translation invariant along $\Span(\cS)$.
\end{proof}

\smallskip

We saw in Example~\ref{example convex and affine} that in the preceding result we cannot replace sublinearity by convexity. However, we may replace sublinearity by $\sigma(\cX,\cX^\ast)$-lower semicontinuity and convexity. In this case, lower semicontinuity forces translation invariance along the $\sigma(\cX,\cX^\ast)$-closure of $\Span(\cS)$ and delivers a dual representation that will be exploited in the context of law-invariant functionals in the next section.

\begin{theorem}
\label{theo: affinity implies translation invariance}
Let $\varphi:\cX\to(-\infty,\infty]$ be proper, convex, and $\sigma(\cX,\cX^\ast)$-lower semicontinuous and $\cS\subset\cX$. If $\varphi$ is affine along every element of $\cS$, then $\varphi$ is translation invariant along $\cM$, where $\cM$ is the $\sigma(\cX,\cX^\ast)$-closure of $\Span(\cS)$. Moreover, for all $Z\in\cM$ and $Y\in\dom(\varphi^\ast)$
\begin{equation}
\label{representation affine piece}
\varphi(Z)=\E_\probp[ZY]+\varphi(0).
\end{equation}
\end{theorem}
\begin{proof}
{\em Step 1}. Take arbitrary $Z\in\cS$ and $Y\in\dom(\varphi^\ast)$. Since $m Z\in\dom(\varphi)$ for every $m\in\R$ by affinity, it follows from Proposition \ref{prop: fenchel moreau} that for every $m\in\R$ we have
\[
\sup_{m\in\R}\{m(\E_\probp[ZY]-\varphi(Z)+\varphi(0))\}-\varphi(0)=\sup_{m\in\R}\{\E_\probp[m ZY]-\varphi(mZ)\} \le\sup_{X\in\cX}\{\E_\probp[XY]-\varphi(X)\}< \infty.
\]
Clearly, this is only possible if $\varphi(Z)=\E_\probp[ZY]+\varphi(0)$. This establishes  \eqref{representation affine piece} when $Z\in\cS$.

\smallskip

{\em Step 2}. Take now arbitrary $Z\in\cS$ and $Y\in\dom(\varphi^\ast)$. It follows from Step 1 that $\E_\probp[ZY]=\varphi(Z)-\varphi(0)=\E_\probp[ZY']$ for every $Y'\in\dom(\varphi^\ast)$. Hence, we infer from Proposition \ref{prop: fenchel moreau} that for every $X\in\cX$
\begin{align*}
\varphi(X+Z)
&=
\sup_{Y'\in\dom(\varphi^\ast)}\{\E_\probp[(X+Z)Y']-\varphi^\ast(Y')\} \\
&=
\sup_{Y'\in\dom(\varphi^\ast)}\{\E_\probp[XY']-\varphi^\ast(Y')\}+\E_\probp[ZY] \\
&=
\varphi(X)+\E_\probp[ZY] \\
&=
\varphi(X)+\varphi(Z)-\varphi(0).
\end{align*}
This shows that $\varphi$ is translation invariant along every element of $\cS$. By Remark~\ref{rem: translation invariance and affinity}, it follows that $\varphi$ is translation invariant along $\Span(\cS)$. In particular, \eqref{representation affine piece} holds also for every $Z\in\Span(\cS)$.

\smallskip

Take now $Z\in\cM$ and let $(Z_\alpha)$ be a net in $\Span(\cS)$ converging to $Z$ and $Y\in\dom(\varphi^\ast)$. Then,
\[
\varphi(Z)\le \liminf_{\alpha}\varphi(Z_\alpha)=\liminf_{\alpha}\E_\probp[Z_\alpha Y] +\varphi(0)=\E_\probp[ZY] +\varphi(0)
\]
by lower semicontinuity at $Z$. Using translation invariance along $\Span(\cS)$ we have for every $\alpha$
\[
\varphi(Z)=\varphi(Z-Z_\alpha)+\varphi(Z_\alpha)-\varphi(0)=\varphi(Z-Z_\alpha)+\E_\probp[Z_\alpha Y].
\]
Hence, by lower semicontinuity at $0$, we easily obtain
\[
\varphi(Z) = \liminf_{\alpha}\E_\probp[Z_\alpha Y]+\liminf_{\alpha}\varphi(Z-Z_\alpha)
\ge \E_\probp[ZY]+\varphi(0).
\]
It follows that $\varphi(Z)=\E_\probp[ZY] +\varphi(0)$ for every $Z\in\cM$. In particular, $\varphi$ is affine on $\cM$. To conclude the proof we may apply what we have showed so far to $\cM$ instead of $\cS$.
\end{proof}

\smallskip

A direct consequence of the preceding result is that when the functional is affine on a set whose linear span is $\sigma(\cX,\cX^\ast)$-dense in $\cX$, it must be affine on the entire space. Its linear part is thus represented by a unique dual element in $\cX^\ast$.

\begin{corollary}
\label{cor: affinity implies translation invariance}
Let $\varphi:\cX\to(-\infty,\infty]$ be proper, convex, and $\sigma(\cX,\cX^\ast)$-lower semicontinuous and $\cS\subset\cX$ such that $\Span(\cS)$ is $\sigma(\cX,\cX^\ast)$-dense in $\cX$. If $\varphi$ is affine along every element of $\cS$, then $\varphi$ is affine on $\cX$ and there exists a unique $Y\in\cX^\ast$ such that for every $X\in\cX$
\[
\varphi(X) = \E_\probp[XY]+\varphi(0).
\]
\end{corollary}


\section{Collapse to the mean}

Throughout this section, we assume that $(\cX,\cX^\ast)$ is a pair satisfying Assumption~\ref{assumption one}. We establish our main result on the ``collapse to the mean'' of convex law-invariant functionals. We start by recalling a well-known result about ``law-invariance equivalence classes''. Here, for every random variable $X\in L^0$ we denote by $q_X$ a fixed quantile function of $X$, i.e.\ a function $q_X:(0,1)\to\R$ satisfying for every $\alpha\in(0,1)$
\[
\inf\{m\in\R \,; \ \probp(X\leq m)\geq\alpha\} \leq q_X(\alpha) \leq \inf\{m\in\R \,; \ \probp(X\leq m)>\alpha\}.
\]

\smallskip

\begin{lemma}
\label{lem: law invariance classes}
For all $X\in\cX$ and $Y\in\cX^\ast$ the set $\cE(X,Y)=\{\E_\probp[X'Y] \,; \ X'\in\cX, \ X'\sim X\}$ is a closed interval such that:
\begin{enumerate}[(i)]
  \item $\inf\cE(X,Y)=\int_0^1q_X(\alpha)q_Y(1-\alpha)d\alpha$.
  \item $\sup\cE(X,Y)=\int_0^1q_X(\alpha)q_Y(\alpha)d\alpha$.
  \item $\cE(X,Y)=\{\E_\probp[XY'] \,; \ Y'\in\cX^\ast, \ Y'\sim Y\}$.
\end{enumerate}
Moreover, $\cE(X,Y)$ is reduced to a singleton if and only if either $X$ or $Y$ is constant.
\end{lemma}
\begin{proof}
It can be proved along the lines of Theorem 9.1 in Luxemburg~\cite{Luxemburg1967} that $\cE(X,Y)$ is a closed interval satisfying assertions {\em (i)} to {\em (iii)}. We refer to Bellini et al.\ \cite{BelliniKochMunariSvindland2020} for a detailed proof. The ``if'' implication in the last assertion is clear. To establish the ``only if'' implication, assume that $\cE(X,Y)$ is reduced to a singleton. In this case, we must have
\begin{align*}
0
&=
\int_0^1q_X(\alpha)q_Y(\alpha)d\alpha-\int_0^1q_X(\alpha)q_Y(1-\alpha)d\alpha \\
&=
\int_0^{1/2}q_X(\alpha)[q_Y(\alpha)-q_Y(1-\alpha)]d\alpha+
\int_{1/2}^1q_X(\alpha)[q_Y(\alpha)-q_Y(1-\alpha)]d\alpha \\
&=
\int_0^{1/2}[q_X(\alpha)-q_X(1-\alpha)][q_Y(\alpha)-q_Y(1-\alpha)]d\alpha.
\end{align*}
Now, assume that either $X$ or $Y$ is not constant. Upon exchanging their roles, we can assume without loss of generality that $X$ is not constant. Then, we find $\beta\in(0,1/2)$ such that $q_X(\alpha)-q_X(1-\alpha)<0$ for almost every $\alpha\in(0,\beta]$. Hence, the above identity can only hold if $q_Y(\alpha)=q_Y(1-\alpha)$ for almost every $\alpha\in(0,\beta]$. Being nondecreasing, $q_Y$ must therefore be almost-surely constant so that $Y$ has to be constant. This delivers the desired implication.
\end{proof}

\smallskip

Recall that, by definition of the topology $\sigma(\cX,\cX^\ast)$, every linear and $\sigma(\cX,\cX^\ast)$-continuous functional $\varphi:\cX\to\R$ can be represented by a suitable $Y\in\cX^\ast$ through the identity $\varphi(X)=\E_\probp[XY]$ for every $X\in\cX$. Hence, it is an immediate consequence of the preceding lemma that any linear and $\sigma(\cX,\cX^\ast)$-continuous functional that is law invariant must ``collapse to the mean''.

\begin{proposition}
\label{prop: collapse under global linearity}
Let $\cM$ be a law-invariant linear subspace of $\cX$ containing a nonconstant random variable. Let $Y\in\cX^\ast$ and consider the linear functional $\varphi:\cM\to\R$ given by $\varphi(X)=\E_\probp[XY]$. The following statements are equivalent:
\begin{enumerate}[(a)]
  \item $\varphi$ is law invariant.
  \item $Y$ is constant.
\end{enumerate}
\end{proposition}

\smallskip

\begin{remark}
Of course, the preceding proposition could be proved directly. It is trivial that {\em (b)} implies {\em (a)}. To see that {\em (a)} implies {\em (b)}, assume $Y$ is not constant so that we find $\alpha\in\R$ satisfying $\probp(Y<\alpha)>0$ as well as $\probp(Y>\alpha)>0$. By nonatomicity, there exist measurable sets $E\subset\{Y<\alpha\}$ and $F\subset\{Y>\alpha\}$ such that $\probp(E)=\probp(F)>0$. Setting $X_1=\one_E$ and $X_2=\one_F$ we see that $X_1$ and $X_2$ belong to $\cX$ and satisfy $X_1\sim X_2$ and $\varphi(X_1) <  \alpha \probp(E)=\alpha\probp(F)< \varphi(X_2)$. This shows that $\varphi$ is not law invariant.
\end{remark}

\smallskip

We now use Lemma~\ref{lem: law invariance classes} to prove that the linear space generated by all the random variables having the same distribution as a given nonconstant random variable with nonzero expectation is $\sigma(\cX,\cX^\ast)$-dense in the space $\cX$. For any random variable $X\in\cX$ set
\[
\cL_X : =\{X'\in\cX \,; \ X'\sim X\}.
\]

\smallskip

\begin{lemma}
\label{lem: density}
For every nonconstant $Z\in\cX$ the following statements hold:
\begin{enumerate}[(i)]
\item If $\E_\probp[Z]\neq0$, then $\Span(\cL_Z)$ is $\sigma(\cX,\cX^\ast)$-dense in $\cX$.
\item If $\E_\probp[Z]=0$, then the $\sigma(\cX,\cX^\ast)$-closure of $\Span(\cL_Z)$ coincides with the set $\{X\in\cX \,; \ \E_\probp[X]=0\}$.
\end{enumerate}
\end{lemma}
\begin{proof}
Let $\cM$ be the $\sigma(\cX,\cX^\ast)$-closure of $\Span(\cL_Z)$. The annihilator of the set $\cM$ is defined by
\[
\cM^\perp:=\{Y\in\cX^\ast \,; \ \forall X\in\cM, \ \E_\probp[XY]=0\}.
\]
Similarly, the annihilator of the set $\cM^\perp$ is given by
\[
\cM^{\perp\perp}:=\{X\in\cX \,; \ \forall Y\in\cM^\perp, \ \E_\probp[XY]=0\}.
\]
Take an arbitrary $Y\in\cM^\perp$. Since $Z$ is not constant and $\{\E_\probp[ZY'] \,; \ Y'\in\cL_Y\} = \{\E_\probp[Z'Y] \,; \ Z'\in\cL_Z\} = \{0\}$ by Lemma~\ref{lem: law invariance classes}, it follows from the same result that $Y$ must be constant. If $\E_\probp[Z]\neq0$, then we must have $Y=0$. In this case, $\cM^\perp=\{0\}$ and it follows from Corollary~5.108 in Aliprantis and Border~\cite{AliprantisBorder2006} that {\em (i)} holds. If $\E_\probp[Z]=0$, then we must have $\cM^\perp=\R$. This implies that $\cM^{\perp\perp}=\{X\in\cX \,; \ \E_\probp[X]=0\}$. Since $\cM=\cM^{\perp\perp}$ by Theorem~5.107 in Aliprantis and Border~\cite{AliprantisBorder2006}, we infer that {\em (ii)} holds.
\end{proof}


\subsubsection*{Affinity along a nonconstant random variable with nonzero expectation}

By combining the previous results we can now easily establish our main result.

\begin{theorem}
\label{theo: collapse}
For a proper, convex, $\sigma(\cX,\cX^\ast)$-lower semicontinuous, law-invariant functional $\varphi:\cX\to(-\infty,\infty]$ the following statements are equivalent:
\begin{enumerate}[(a)]
    \item The functional $\varphi$ is affine along a nonconstant $Z\in\cX$ with $\E_\probp[Z]\neq0$.
    \item The functional $\varphi$ is translation invariant along a nonconstant $Z\in\cX$ with $\E_\probp[Z]\neq0$.
	\item There exists $a\in\R$ such that $\varphi(X)=a\E_\probp[X]+\varphi(0)$ for every $X\in\cX$.
\end{enumerate}
\end{theorem}
\begin{proof}
It follows from Theorem~\ref{theo: affinity implies translation invariance} that {\em (a)} and {\em (b)} are equivalent. To conclude, we only have to show that {\em (a)} implies {\em (c)}. To this effect, assume that $\varphi$ is affine along a nonconstant random variable $Z\in\cX$ with $\E_\probp[Z]\neq0$. Note that, by Lemma~\ref{lem: density}, the $\sigma(\cX,\cX^\ast)$-closure of $\Span(\cL_Z)$ is $\cX$. Noting that, by law invariance, $\varphi$ is affine along each element of $\cL_Z$, we can apply Corollary~\ref{cor: affinity implies translation invariance} to obtain that
\[
\varphi(X)=\E_\probp[XY]+\varphi(0)
\]
for all $X\in\cX$ and $Y\in\dom(\varphi^\ast)$. It now suffices to apply Proposition~\ref{prop: collapse under global linearity} to the functional $\varphi-\varphi(0)$ to infer that $Y$ must be constant and conclude the proof.
\end{proof}

\smallskip

\begin{remark}
We show that lower semicontinuity is necessary for the above ``collapse to the mean'' to hold. Let $\cA=\{X\in L^1 \,; \ \mbox{$X$ has a discrete distribution}\}$ and define $\varphi:L^1\to(-\infty,\infty]$ by
\[
\varphi(X)=
\begin{cases}
0 & \mbox{if} \ X\in\cA,\\
\infty & \mbox{otherwise}.
\end{cases}
\]
It is clear that $\varphi$ is convex and law invariant. Moreover, for every event $E\in\cF$ with $\probp(E)\in(0,1)$ we have that $\varphi$ is linear (in fact, null) on the vector space spanned by the nonconstant random variable $Z=\one_E$. However, $\varphi$ fails to be $\sigma(L^1,L^\infty)$-lower semicontinuous. To see this, take a positive random variable $X\in\cX\setminus\cA$. Then, we can always find an increasing sequence $(X_n)\subset\cA$ such that $X_n\to X$ almost surely. It follows from the Dominated Convergence Theorem that $X_n\to X$ with respect to $\sigma(L^1,L^\infty)$ but
\[
\varphi(X) = \infty > 0 = \liminf_{n\to\infty}\varphi(X_n),
\]
showing that $\varphi$ is not $\sigma(L^1,L^\infty)$-lower semicontinuous.
\end{remark}


\subsubsection*{Affinity along a nonconstant random variable with zero expectation}

If the random variable along which a functional is affine has zero expectation, then the functional is simply the composition of a convex real function and the expectation functional.

\begin{theorem}
\label{theo: collapse zero expectation}
For a proper, convex, $\sigma(\cX,\cX^\ast)$-lower semicontinuous, law-invariant functional $\varphi:\cX\to(-\infty,\infty]$ the following statements are equivalent:
\begin{enumerate}[(a)]
    \item The functional $\varphi$ is affine along a nonconstant $Z\in\cX$ with $\E_\probp[Z]= 0$.
    \item The functional $\varphi$ is translation invariant along a nonconstant $Z\in\cX$ with $\E_\probp[Z]=0$.
    \item $\varphi(X)=\varphi(\E_\probp[X])$ for every $X\in\cX$.
\end{enumerate}
\end{theorem}
\begin{proof}
It follows from Theorem~\ref{theo: affinity implies translation invariance} that {\em (a)} and {\em (b)} are equivalent. To conclude, we only have to show that {\em (a)} implies {\em (c)}. Hence, assume that $\varphi$ is affine along a nonconstant $Z\in\cX$ with $\E_\probp[Z]=0$. Let $\cM=\{X\in\cX \,; \ \E_\probp[X]=0\}$, which by Lemma~\ref{lem: density} is the $\sigma(\cX,\cX^\ast)$-closure of $\Span(\cL_Z)$. By Theorem~\ref{theo: affinity implies translation invariance},
\[
\varphi(X)=\E_\probp[XY]+\varphi(0)
\]
for all $X\in\cM$ and $Y\in\dom(\varphi^\ast)$. It follows from Proposition~\ref{prop: collapse under global linearity} that $Y$ must be constant. Hence,
\[
\varphi(X) = \varphi(\E_\probp[X])+\varphi(X-\E_\probp[X])-\varphi(0) = \varphi(\E_\probp[X])+\varphi(0)-\varphi(0) = \varphi(\E_\probp[X])
\]
by translation invariance along $\cM$. This delivers the desired implication.
\end{proof}

\smallskip

Although, in general, there is no full ``collapse to the mean'' if the functional is affine along a direction with zero expectation, we do obtain a full ``collapse to the mean'' as soon as we additionally have translation invariant along constant random variables. This is a situation that is often encountered in applications.

\begin{corollary}
\label{cor: collapse}
For a proper, convex, $\sigma(\cX,\cX^\ast)$-lower semicontinuous, law-invariant functional $\varphi:\cX\to(-\infty,\infty]$ that is translation invariant along $1$ the following statements are equivalent:
\begin{enumerate}[(a)]
    \item The functional $\varphi$ is affine along a nonconstant $Z\in\cX$.
    \item The functional $\varphi$ is translation invariant along a nonconstant $Z\in\cX$.
	\item There exists $a\in\R$ such that $\varphi(X)=a\E_\probp[X]+\varphi(0)$ for every $X\in\cX$.
\end{enumerate}
\end{corollary}
\begin{proof}
If $\E_\probp[Z]\neq 0$, then the equivalences follow from Theorem~\ref{theo: collapse}. If $\E_\probp[Z]=0$, it suffices to show that {\em (a)} implies {\em (c)} due to Theorem~\ref{theo: collapse zero expectation}. In this case, the same result implies that $\varphi(X)=\varphi(\E_\probp[X])$ for every $X\in\cX$ whenever {\em (a)} holds. Then, by translation invariance along $1$, there exists $a\in\R$ such that $\varphi(X)=\varphi(0)+a\E_\probp[X]$ for every $X\in\cX$.
\end{proof}


\section{Applications}

In this final section we point out connections to other works in the literature in which a ``collapse to the mean'' was established. We also highlight some applications of the ``collapse to the mean'' to pricing functionals and risk measures. Throughout the entire section we continue to denote by $(\cX,\cX^\ast)$ a pair satisfying Assumption~\ref{assumption one}.


\subsubsection*{Collapse to the mean in the literature}

We now show how to derive the known ``collapse to the mean'' results of the literature from our general results. We start with the results in Castagnoli et al. \cite{CastagnoliMaccheroniMarinacci2004} who focus on law-invariant Choquet integrals on $L^\infty$. Recall that a set function $c:\cF\to[0,1]$ is called a \textit{submodular capacity}\footnote{We prefer this terminology to ``submodular nonadditive probability'', which is used in Castagnoli et al. \cite{CastagnoliMaccheroniMarinacci2004}.} or if it satisfies the following conditions:
\begin{enumerate}[{\rm (1)}]
  \item $c(\Omega)=1$ and $c(E)=0$ for every $E\in\cF$ such that $\probp(E)=0$.
  \item $c(E)\leq c(F)$ for all $E,F\in\cF$ such that $E\subset F$.
  \item $c(E_n)\to 0$ for every decreasing sequence $(E_n)\subset\cF$ such that $\bigcap_{n\in\N} E_n=\emptyset$.
  \item $c(E\cup F)\leq c(E)+c(F)-c(E\cap F)$ for all $E,F\in\cF$.
\end{enumerate}
The Choquet integral associated to a submodular capacity $c$ is the functional $\E_c:L^\infty\to\R$ defined by
\[
\E_c[X] := \int_{-\infty}^0(c(X>x)-1)dx+\int_0^\infty c(X>x)dx.
\]
The ``collapse to the mean'' says that a Choquet integral associated with a submodular capacity $c$ reduces to the standard expectation under $\probp$ whenever it is law invariant under $\probp$ and linear along a nonconstant random variable.

\begin{theorem}[Theorem 3.1 in \cite{CastagnoliMaccheroniMarinacci2004}]
\label{theo: castagnoli et al}
Let $c$ be a submodular capacity. If $\E_c$ is law invariant under $\probp$ and $\E_c[-Z]=-\E_c[Z]$ for a nonconstant $Z\in L^\infty$, then $\E_c[X]=\E_\probp[X]$ for every $X\in L^\infty$ or equivalently $c=\probp$.
\end{theorem}
\begin{proof}
It is clear that $\E_c$ is proper. It follows from Schmeidler \cite{Schmeidler1986} that $\E_c$ is sublinear and translation invariant along $1$. Then, $\E_c$ is automatically (Lipschitz) continuous with respect to the $L^\infty$ norm by Lemma 4.3 in F\"{o}llmer and Schied \cite{FoellmerSchied2011}. This implies that, being law invariant, $\E_c$ is $\sigma(L^\infty,L^1)$-lower semicontinuous by Example~\ref{example: fatou}. Since $\E_c$ is affine along $Z$ by assumption, we infer from Corollary~\ref{cor: collapse} that there exists $a\in\R$ such that $\E_c[X]=a\E_\probp[X]+\E_c[0]$ for every $X\in L^\infty$. We conclude by observing that $\E_c[0]=0$ by sublinearity and $a=\E_c[1]=1$.
\end{proof}

\smallskip

The preceding result can be recast as a ``collapse to the mean'' for comonotonic functionals on $L^\infty$. Recall that a functional $\varphi:L^\infty\to\R$ is {\em comonotonic} if $\varphi(X+Y)=\varphi(X)+\varphi(Y)$ for all comonotone random variables $X,Y\in L^\infty$.

\begin{corollary}
\label{cor: comonotonic pricing}
Let $\varphi:L^\infty\to\R$ be a sublinear, increasing, and comonotonic functional satisfying
\begin{equation}
\label{comonotonic translation invariance}
\varphi(X+m) = \varphi(X)+m
\end{equation}
for all $X\in L^\infty$ and $m\in\R$. If $\varphi$ is law invariant and $\varphi(-Z)=-\varphi(Z)$ for a nonconstant $Z\in L^\infty$, then $\varphi(X)=\E_\probp[X]$ for every $X\in L^\infty$
\end{corollary}
\begin{proof}
It follows from the classical results in Schmeidler \cite{Schmeidler1986}, see also Wang et al.~\cite{WangYoungPanjer1997}, that a sublinear, increasing, and comonotonic functional satisfying the translation invariance property \eqref{comonotonic translation invariance} can be represented as a Choquet integral with respect to a submodular capacity. The claim is then a direct consequence of Theorem \ref{theo: castagnoli et al}.
\end{proof}

\smallskip

The focus of Frittelli and Rosazza Gianin \cite{FrittelliRosazza2005} is on law-invariant convex risk measures on $L^\infty$. Their ``collapse to the mean'' extends the previous results from the literature by showing that a law-invariant convex risk measure on $L^\infty$ reduces to (the negative of) a standard expectation under the reference probability measure whenever the risk measure is linear along a nonconstant random variable.

\begin{theorem}[Proposition 9 in \cite{FrittelliRosazza2005}]
\label{theo: frittelli and rosazza}
Let $\varphi:L^\infty\to\R$ be a convex decreasing functional satisfying
\[
\varphi(X+m) = \varphi(X)-m
\]
for all $X\in L^\infty$ and $m\in\R$. If $\varphi$ is law invariant and there exists a nonconstant $Z\in L^\infty$ such that $\varphi(mZ)=m\varphi(Z)$ for every $m\in\R$, then $\varphi(X)=-\E_\probp[X]$ for every $X\in L^\infty$.
\end{theorem}
\begin{proof}
By assumption, $\varphi$ is translation invariant along $1$ and affine along $Z$. In particular, $\varphi$ is automatically (Lipschitz) continuous with respect to the $L^\infty$ norm by Lemma 4.3 in F\"{o}llmer and Schied \cite{FoellmerSchied2011}. This implies that, being law invariant, $\varphi$ is $\sigma(L^\infty,L^1)$-lower semicontinuous by Example~\ref{example: fatou}. As a result, we infer from Corollary~\ref{cor: collapse} that there exists $a\in\R$ such that $\varphi(X)=a\E_\probp[X]+\varphi(0)$ for every $X\in L^\infty$. We conclude by observing that $\varphi(0)=0$ and $a=\varphi(1)=-1$.
\end{proof}

\smallskip

We close this section by highlighting the three major differences between our results and those in \cite{CastagnoliMaccheroniMarinacci2004} and \cite{FrittelliRosazza2005}:

\begin{enumerate}[(1)]
  \item Instead of working only with bounded random variables, our model space is allowed to belong to  a wide class of spaces of integrable random variables containing the space of bounded random variables. This covers all the standard model spaces encountered in the literature and ensures the broad applicability of our results.

  \item We establish that the ``collapse to the mean'' remains valid for a larger class of law-invariant functionals by either dropping or weakening the following assumptions required in the literature: monotonicity, translation invariance along constant random variables, (Lipschitz) continuity, positive homogeneity, and comonotonicity.

  \item The proof in Castagnoli et al. \cite{CastagnoliMaccheroniMarinacci2004} is based on probabilistic arguments and tailored to Choquet integrals. The proof in Frittelli and Rosazza Gianin \cite{FrittelliRosazza2005} is obtained through a careful manipulation of the Kusuoka representation of convex risk measures on $L^\infty$ so that extending their approach to our setting would require to first establish a Kusuoka representation on general spaces of random variables. The strategy used in this paper relies solely on a direct analysis of the link between the two key concepts under investigation --- law invariance and linearity --- and does not require preliminary structural results about law-invariant functionals. The key observation is that the vector space generated by the random variables having the same distribution as a nonconstant random variable $Z$ (with nonzero expectation) is dense in the underlying model space. As a result, linearity along $Z$ together with law invariance forces linearity on a dense subspace. This, in turn, implies linearity on the entire space by lower semicontinuity. We believe that our strategy is rather intuitive and sheds new light on the structure of law invariance and its relationship with linear and topological structures.
\end{enumerate}


\subsubsection*{Law-invariant pricing rules}

The pricing of insurance contracts is one of the key topics in actuarial science. The classical approach based on expected utility theory is thoroughly presented in standard textbooks such as B\"{u}hlmann \cite{Buehlmann1970}, Borch \cite{Borch1974}, Gerber \cite{Gerber1979}. Since the pioneering contributions of these authors, it has become customary in the theoretical literature to address the pricing problem in an ``axiomatic'' way by prescribing a set of economically plausible requirements that a ``good'' pricing rule should satisfy. An early survey of the axiomatic approach to insurance pricing can be found in Goovaerts et al.\ \cite{GoovaertsDeVylderHaezendonck1984} and Deprez and Gerber \cite{DeprezGerber1985}. An updated picture is presented in Laeven and Goovaerts \cite{LaevenGoovaerts2014}. In a pricing setting, the elements of $\cX$ are interpreted as the payoffs of financial contracts at a given future date. A payoff is called {\em risk free} whenever it is constant and {\em risky} otherwise. A pricing rule assigns to each payoff its (buying) price.

\begin{definition}
A pricing rule is a functional $\pi:\cX\to(-\infty,\infty]$ satisfying $\pi(0)=0$. A payoff $X\in\cX$ is {\em frictionless (under $\pi$)} if it satisfies the following conditions:
\begin{enumerate}[(1)]
  \item $\pi(-X)=-\pi(X)$.
  \item $\pi(\lambda X)=\lambda\pi(X)$ for every $\lambda\in(0,\infty)$.
\end{enumerate}
\end{definition}

\smallskip

For every $X\in\cX$ the quantity $\pi(X)-(-\pi(-X))$ can be interpreted as the difference between the buying and the selling price of $X$, i.e.\ as the ``bid-ask spread'' of $X$; see e.g.\ Jouini \cite{Jouini2000}. A payoff is frictionless precisely when its bid-ask spread is zero and the price per unit does not depend on the transacted volume.

\smallskip

The ``collapse to the mean'' recorded in Theorem \ref{theo: castagnoli et al} was originally formulated in the context of Choquet pricing. In view of Corollary \ref{cor: comonotonic pricing}, that result can be equivalently formulated as follows: {\em The expectation under the reference probability measure $\probp$ is the only law-invariant, sublinear, increasing, comonotonic pricing functional on $L^\infty$ under which every risk-free payoff and some risky payoff are frictionless}. As a direct consequence of Theorem~\ref{theo: collapse} we obtain the following generalization of this result.

\begin{proposition}
\label{prop: collapse literature 1}
Let $\pi$ be a proper, convex, $\sigma(\cX,\cX^\ast)$-lower semicontinuous, law-invariant pricing rule. If some risky payoff $Z\in\cX$ with $\E_\probp[Z]\neq0$ is frictionless under $\pi$, then there exists $a\in\R$ such that
\[
\pi(X)=a\E_\probp[X]
\]
for every $X\in\cX$. In particular, every payoff is frictionless under $\pi$. (The condition $\E_\probp[Z]\neq0$ can be removed if the risk-free payoff $1$ is frictionless under $\pi$).
\end{proposition}
\begin{proof}
Take any payoff $Z\in\cX$ and note first that $Z$ is frictionless if $\pi$ is linear along it. The converse also holds. Indeed, if $Z$ is frictionless, then for every $m\in\R$ we have $\pi(mZ)=m\pi(Z)$ whenever $m\geq0$ (recall that $\pi(0)=0$ by our initial assumption on $\pi$) and
\[
\pi(mZ) = \pi(-(-m)Z) = -m\pi(-Z) = m\pi(Z)
\]
whenever $m<0$. The desired statements now follow directly from Theorem~\ref{theo: collapse} and Corollary~\ref{cor: collapse}.
\end{proof}

\smallskip

The preceding result extends the message of Castagnoli et al.\ \cite{CastagnoliMaccheroniMarinacci2004} beyond their bounded-payoff setting and beyond law-invariant Choquet integrals: In a market where there exists at least one frictionless risky payoff, no reasonable convex and lower semicontinuous pricing rule can be law invariant. In particular, this shows that the law-invariant pricing rules put forward in Wang \cite{Wang2000,Wang2002}, which involve unbounded payoffs, cannot be expected to harmonize insurance and derivatives pricing. We note that the more recent literature on market-consistent valuation (see e.g.\ Malamud et al. \cite{MalamudTrubowitzWuethrich2008}, Pelsser and Stadje \cite{PelsserStadje2014}, Dhaene et al.\ \cite{DhaeneStassenBarigouLindersChen2017}) seems to be, at least implicitly, aware of this limitation and requires only partial law invariance, e.g.\ for payoffs that depend on pure insurance risk only. We also refer to the economic premium principles in B\"{u}hlmann \cite{Buehlmann1980} and B\"{u}hlmann \cite{Buehlmann1984} for early examples of premium principles that are not law invariant on the entire reference payoff space and to Deprez and Gerber \cite{DeprezGerber1985} for a first systematic treatment of premium principles beyond law invariance. From this perspective, our result provides a rigorous justification of why law invariance cannot be stipulated when pricing the entire universe of financial contracts.


\subsubsection*{Law-invariant risk measures based on general eligible assets}

The paper by Artzner et al. \cite{ArtznerDelbaenEberHeath1999} has been a landmark contribution in the theory of risk measures. In a regulatory context, a risk measures assign the  minimal amount of capital that has to be raised and invested in a fixed financial asset, called the \textit{eligible asset}, to ensure an acceptable profit-and-loss profile. The acceptability criterion is pre-specified by the regulator. In the literature, it is standard to assume that the eligible asset is frictionless in the sense that it is available in arbitrary quantities and its price per unit does not depend on the transacted volume. In this case, the corresponding risk measures are naturally translation invariant as recalled below. In the context of risk measures, the elements of $\cX$ are interpreted as (net) capital positions of financial firms at a fixed future date.

\begin{definition}
A {\em (frictionless) eligible asset} is a couple $S=(S_0,S_1)$ with strictly-positive price $S_0\in\R$ and nonzero positive payoff $S_1\in\cX$. We say that $S$ is {\em risk free} if $S_1$ is constant and {\em risky} otherwise. We say that $S$ is {\em cash} if $S=(1,1)$. A functional $\rho:\cX\to(-\infty,\infty]$ is said to be an {\em $S$-additive risk measure} if it satisfies the following properties:
\begin{enumerate}[(1)]
  \item $\rho(X+mS_1)=\rho(X)-mS_0$ for all $X\in\cX$ and $m\in\R$.
  \item $\rho$ is decreasing.
\end{enumerate}
When $S$ is cash, we speak of cash-additivity instead of $S$-additivity.
\end{definition}

\smallskip

It is well known that, for every $X\in\cX$, an $S$-additive risk measure can always be expressed as
\[
\rho(X) = \inf\left\{m\in\R \,; \ X+\frac{m}{S_0}S_1\in\cA_\rho\right\},
\]
where $\cA_\rho = \{X\in\cX \,; \ \rho(X)\leq0\}$. The set $\cA_\rho$ consists of all the capital positions that are deemed acceptable from a regulatory perspective. Hence, for every position $X\in\cX$, the quantity $\rho(X)$ can be interpreted as the minimum amount of capital that has to be raised and invested in the eligible asset to ensure acceptability. This type of risk measures has been thoroughly investigated in the case of a cash eligible asset; see e.g.\ F\"{o}llmer and Schied \cite{FoellmerSchied2011}. The case of a general eligible asset has been studied, e.g., in Artzner et al.\ \cite{ArtznerDelbaenKoch2009} and Farkas et al.\ \cite{FarkasKochMunari2014a,FarkasKochMunari2014}.

\smallskip

There are many examples of law-invariant risk measures when the eligible asset is risk free. One question is whether law invariance can hold when the eligible asset is risky. This question was taken up in a bounded setting in Frittelli and Rosazza Gianin \cite{FrittelliRosazza2005}. A slight reformulation of Theorem \ref{theo: frittelli and rosazza} reads as follows: {\em The expectation under the reference probability measure $\probp$ is, up to a sign, the only law-invariant, convex, cash-additive risk measure on $L^\infty$ that is $S$-additive for a risky eligible asset $S$ and assigns the value $0$ to the zero position}. As an application of our general ``collapse to the mean'' we obtain the following generalization of this result.

\begin{proposition}
\label{prop: single asset}
Let $\rho$ be a proper, convex, $\sigma(\cX,\cX^\ast)$-lower semicontinuous, law-invariant, $S$-additive risk measure such that $\rho(0)<\infty$. If the eligible asset $S$ is risky, then for every $X\in\cX$
\[
\rho(X)=\frac{S_0}{\E_\probp[S_1]}\E_\probp[-X]+\rho(0).
\]
(If $\rho$ is cash-additive, then $\E_\probp[S_1]=S_0$).
\end{proposition}
\begin{proof}
Since $\rho(0)\in\R$ and $\rho$ is an $S$-additive risk measure, we have that $\rho$ is translation invariant and, hence, affine along the payoff $S_1$. As $S_1$ is nonconstant and satisfies $\E_\probp[S_1]>0$, it follows from Theorem~\ref{theo: collapse} that there exist $a,b\in\R$ such that $\rho(X) = a\E_\probp[X]+b$ for every $X\in\cX$. We infer that $b=\rho(0)$ and
\[
a = \frac{\rho(S_1)-\rho(0)}{\E_\probp[S_1]} = -\frac{S_0}{\E_\probp[S_1]}.
\]
If $\rho$ is also cash-additive, then $a+b=\rho(1)=\rho(0)-1=b-1$, showing that $\E_\probp[S_1]=S_0$.
\end{proof}


\subsubsection*{Relevant cash-based risk measures}

We conclude the section on applications by showing that our general ``collapse to the mean'' does not only deliver ``non-existence'' statements but can also be exploited to derive ``positive'' results. We focus on cash-additive risk measures satisfying suitable relevance properties.

\begin{definition}
We say that $\rho:\cX\to(-\infty,\infty]$ is {\em relevant} if for every $X\in\cX$ we have
\[
X\geq0, \ \probp(X>0)>0 \ \implies \ \rho(-X)>0
\]
and {\em strongly relevant} if for every $X\in\cX$ we have
\[
X\neq0, \ \rho(X)\leq0 \ \implies \ \rho(-X)>0.
\]
\end{definition}

\smallskip

Note that a strongly-relevant functional that is decreasing and satisfies $\rho(0)\leq0$ is also relevant. The property of relevance, which is sometimes known under the name of {\em sensitivity}, has been studied, e.g., in Stoica \cite{Stoica2006} and F\"{o}llmer and Schied \cite{FoellmerSchied2011} in connection with generalized no-arbitrage conditions.

\smallskip

Our ``collapse to the mean'' can be used to show that, with the exception of the negative of the expectation, every cash-additive risk measure that is sublinear, lower semicontinuous, and law invariant is automatically strongly relevant. In particular, this implies that every risk measure of the above type is always relevant.

\begin{proposition}
Let $\rho$ be a sublinear, $\sigma(\cX,\cX^\ast)$-lower semicontinuous, law-invariant, cash-additive risk measure. Then, one of the following two alternatives holds:
\begin{enumerate}[(i)]
  \item $\rho(X)=\E_\probp[-X]$ for every $X\in\cX$.
  \item $\rho$ is strongly relevant.
\end{enumerate}
In particular, $\rho$ is always relevant.
\end{proposition}
\begin{proof}
Assume that $\rho$ is not strongly relevant. Then, we must find a nonzero $Z\in\cX$ such that $\rho(Z)\leq0$ as well as $\rho(-Z)\leq0$. As $\rho$ is sublinear, we also have
\[
0 = \rho(0) = \rho(Z-Z) \leq \rho(Z)+\rho(-Z).
\]
This implies that $\rho(-Z)=-\rho(Z)$. But then $\rho(mZ)=m\rho(Z)$ for every $m\in\R$ again by sublinearity, showing that $\rho$ is linear along $Z$. Note that $Z$ cannot be constant for otherwise
\[
0 \leq -\rho(Z) = Z = \rho(-Z) \leq 0
\]
would imply that $Z=0$. As a result of Corollary~\ref{cor: collapse}, there must exist $a,b\in\R$ such that $\rho(X)=a\E_\probp[X]+b$ for every $X\in\cX$. To conclude, it suffices to note that $b=\rho(0)=0$ and $a=\rho(1)=-1$.
\end{proof}

\smallskip

\begin{remark}
The preceding result does not generally hold if $\rho$ is only assumed to be convex. To see this, define $\rho:L^1\to(-\infty,\infty]$ by setting
\[
\rho(X) = \inf\{m\in\R \,; \ \E_\probp[\min(X+m,0)]\geq-1\}.
\]
It is immediate to verify that $\rho$ is a convex, $\sigma(L^1,L^\infty)$-lower semicontinuous, law-invariant, cash-additive risk measure. However, we have $\rho(-1)=0$, showing that $\rho$ is neither relevant nor strongly relevant.
\end{remark}


\section*{Acknowledgments}

Partial support through the SNF project 100018-189191 ``Value
Maximizing Insurance Companies: An Empirical Analysis of the Cost of Capital and Investment
Policies'' is gratefully acknowledged.


\end{document}